\documentclass{birkjour}
\usepackage{amsmath, tikz,amssymb,color,cite,pifont}
\usetikzlibrary{arrows,decorations.markings}

\newtheorem{thm}{Theorem}[section]
\newtheorem{cor}[thm]{Corollary}
\newtheorem{lem}[thm]{Lemma}
\newtheorem{prop}[thm]{Proposition}
\theoremstyle{definition}
\newtheorem{ex}[thm]{Example}
\newtheorem{defn}[thm]{Definition}
\theoremstyle{remark}

\newcommand{\A}{\mathcal{A}}

\newcommand{\I}{\mathfrak{I}}
\newcommand{\M}{\mathfrak{M}}

\newcommand{\Z}{\mathbb{Z}}

\renewcommand{\O}{\mathcal{O}}

\newcommand{\gn}{\operatorname{gn}}
\newcommand{\bb}{\operatorname{b}} 

\tikzstyle{ball} = [circle, shading=ball, ball color=red!80!white, minimum size=0.4cm]

\title{The Information Flow Problem on Clock Networks}

\author[Atkins]{Ross Atkins}
\address{
University of Oxford \\
Department of Statistics \\
1 South Parks Road \\
Oxford OX1 3TG \\
United Kingdom}
\email{ross.atkins@univ.ox.ac.uk}

\begin{document}

\begin{abstract}
The information flow problem on a network asks whether $r$ senders, $v_1,v_2, \ldots ,v_r$ can each send messages to $r$ corresponding receivers $v_{n+1}, \ldots ,v_{n+r}$ via intermediate nodes $v_{r+1}, \ldots ,v_n$. For a given finite $R \subset \Z^+$, the clock network $N_n(R)$ has edge $v_iv_k$ if and only if $k>r$ and $k-i \in R$. We show that the information flow problem on $N_n(\{1,2, \ldots ,r\})$ can be solved for all $n \geq r$. We also show that for any finite $R$ such that $\gcd(R)=1$ and $r = \max(R)$, we show that the information flow problem can be solved on $N_n(R)$ for all $n \geq 3r^3$. This is an improvement on the bound given in \cite{Wu2009guessing} and answers an open question from \cite{Riis2007information}.
\end{abstract}


\maketitle

\section{The information Flow Problem}

The information flow problem (Definition~\ref{defn:ifp}) is an important problem for multiuser information theory. This problem was introduced in \cite{Ahlswede2000network} to formalise the multiple unicast problem. It was shown that the information flow problem is equivalent to the guessing number of a related digraph \cite{Riis2007information}. The same paper poses an open question regarding the guessing number of a class of digraphs known as clock digraphs (Definition~\ref{defn:clock_digraph}). Corollary~\ref{cor:large_n} answers this question. 

\begin{defn}
\label{defn:ifn}
A \emph{network} of length $n$ and width $r$ is an acyclic digraph $N$ with vertex set $\{ v_i \}_{i=1}^{n+r}$ such that the \emph{input nodes} (vertices $v_1,v_2, \ldots ,v_r$) have no incoming edges. Vertices $v_{n+1},v_{n+2}, \ldots ,v_{n+r}$ are called the \emph{output nodes} and vertices $v_{r+1},v_{r+2}, \ldots , v_n$ are called \emph{intermediate nodes}. For any $r < k \leq n+r$, let $\Gamma(k)$ denote set of all indices, $i$, such that $v_iv_k$ is an edge.
\end{defn}

For any positive integer $m$, let $[m]$ denote the set $\{ 1,2,3, \ldots ,m\}$. 

\begin{defn}
\label{defn:circuit}
For any network $N$ and any integer $s \geq 2$, a \emph{circuit} on $N$ over $\Z_s$, is a $n$-tuple of functions $F = (f_{r+1},f_{r+2}, \ldots ,f_{n+r})$, 
	$$ f_k : \Z_s^{\Gamma(k)} \rightarrow \Z_s \qquad \forall \; r < k \leq n+r, $$
where $n$ and $r$ are the length and width respectively of $N$. For each input $c = (c_1,c_2, \ldots ,c_r) \in \Z_s^r$, let $X = (X_1,X_2, \ldots ,X_{n+r})$ denote the unique $(n+r)$-tuple in $\Z_s^{n+r}$ such that $X_i = c_i$ for all $i \in [r]$ and 
	$$ X_k = f_k \left( X_i \: \big| \: i \in \Gamma(k) \right) \qquad \forall \; r < k \leq n+r. $$
$X$ is called the \emph{valuation} of $F$. 
\end{defn}

\begin{defn}
\label{defn:M_F}
A circuit, $F = (f_{r+1},f_{r+2}, \ldots ,f_{n+r})$, is called \emph{linear} if and only if each function $f_k$ is a linear map. For any linear circuit $F$, let $M_F$ denote the \emph{$R$-circuit matrix} of $F$; the linear map $M_F : \Z_s^r \rightarrow \Z_s^{n+r}$ such that $X = M_F(c)$ for all inputs $c \in \Z_s^r$. \emph{i.e.} $M_F$ is a $(n+r) \times r$ matrix such that 
	$$ X^T = M_F c^T $$
where $X^T$ and $c^T$ are the column vectors of the valuation $X$ and the input $c$ respectively. The first $r$ rows of the $R$-circuit matrix $M_F$ are a copy of the $r \times r$ identity matrix, $I_r$. 
\end{defn}

\begin{defn}
\label{defn:ifp} 
A network $N$ of width $r$ is $s$-\emph{solvable} if and only if there exists a circuit on $N$ over $\Z_s$ such that for all inputs $c \in \Z_s^r$, the valuation satisfies 
	$$ (X_1,X_2, \ldots ,X_r) = c = (X_{n+1},X_{n+2}, \ldots ,X_{n+r}) $$ 
A network $N$ of width $r$ is \emph{linearly $s$-solvable} if and only if there exists a linear circuit $F$ on $N$ over $\Z_s$ such that the final $r$ rows of $M_F$ are a copy of $I_r$. For a given network $N$ and an integer $s \geq 2$, the \emph{Information Flow Problem} asks whether or not $N$ is $s$-solvable. Similarly, the \emph{Linear Information Flow Problem} asks whether or not $N$ is linearly $s$-solvable.
\end{defn}

It is natural to consider the information flow problem as an information theory problem in the following way. Each input node, $v_i$, is a sender trying to send a message to its corresponding receiver at node $v_{n+i}$ via the network of internal nodes. The elements of the group $\Z_s$ correspond to the $s$ distinct possible messages that could be sent along each edge. There is a traditional method for solving the information flow problem, called ``routing'', in which each intermediate node simply passes on one of the messages it receives. A network can only be solved by routing if and only if there exist vertex disjoint paths from each sender to its corresponding receiver. There are many examples in which a network is solvable, but cannot be solved by routing alone \cite{Cannons2006network,Li2003linear}. Instead we allow each non-input node, $v_k$, to perform some function, $f_k$, on the messages it receives from nodes $\Gamma(k)$. Each node $v_i$ must send the same message to all nodes $v_k$ such that $i \in \Gamma(k)$. Linear circuits are of interest because they are fast to compute and linear circuits are sufficient to solve a large family of networks (Theorems~\ref{thm:full_clock} and~\ref{thm:large_n}). 

The information flow problem also has an application to computing the guessing number \cite{Christofides2011guessing,Riis2006utilising} and the information defect \cite{Alon2008broadcasting,Gadouleau2011graph,Riis2007graph} of directed graphs. Specifically, for any network $N$ with input nodes $v_1,v_2, \ldots ,v_r$ and output nodes $v_{n+1},v_{n+2}, \ldots ,v_{n+r}$, let $G_N$ denote the digraph obtained by identifying vertex $v_i$ with $v_{n+i}$ for all $1 \leq i \leq r$. The relationship between the $s$-solvability of a network $N$ and the guessing number (and information defect) of $G_N$ is presented in Theorem~\ref{thm:information_defect} which originally appears in \cite{Riis2007information}. Note that for our purposes it does not matter if edge $v_{n+i}v_k$ is replaced with $v_iv_k$ (nor would it make any difference if both edges were included) because, for a circuit which solves the network, the valuation would satisfy $X_i=X_{n+i}$. 

\begin{figure}
\begin{center}
\begin{tikzpicture}
\tikzset{myptr/.style={decoration={markings,mark=at position 1 with 
    {\arrow[scale=2,>=stealth]{>}}},postaction={decorate}}}
\draw[fill=black!10,black!10] (1.5,0) circle (2);
\draw[fill=black!10,black!10] (9,0) circle (2);
\draw[myptr] (3,0) -- (4.3,0);
\draw[myptr] (4.5,0) -- (5.8,0);
\draw[myptr] (6,0) -- (7.3,0);
\draw[myptr] (7.5,0) -- (8.8,0);
\draw[myptr] (9,0) -- (10.3,0);
\draw[myptr] (2.25,1.5) to [out=0,in=120] (4.4,0.2);
\draw[myptr] (3.75,1.5) to [out=0,in=120] (5.9,0.2);
\draw[myptr] (5.25,1.5) to [out=0,in=120] (7.4,0.2);
\draw[myptr] (6.75,1.5) to [out=0,in=120] (8.9,0.2);
\draw[myptr] (8.25,1.5) to [out=0,in=120] (10.4,0.2);
\draw (0,0) to [out=60,in=180] (2.25,1.5);
\draw (1.5,0) to [out=60,in=180] (3.75,1.5);
\draw (3,0) to [out=60,in=180] (5.25,1.5);
\draw (4.5,0) to [out=60,in=180] (6.75,1.5);
\draw (6,0) to [out=60,in=180] (8.25,1.5);
\node[style=ball] at (0,0) {};
\node[style=ball] at (1.5,0) {};
\node[style=ball] at (3,0) {};
\node[style=ball] at (4.5,0) {};
\node[style=ball] at (6,0) {};
\node[style=ball] at (7.5,0) {};
\node[style=ball] at (9,0) {};
\node[style=ball] at (10.5,0) {};
\node [below] at (0,-0.15) {$v_1$};
\node [below] at (1.5,-0.15) {$v_2$};
\node [below] at (3,-0.15) {$v_3$};
\node [below] at (4.5,-0.15) {$v_4$};
\node [below] at (6,-0.15) {$v_5$};
\node [below] at (7.5,-0.15) {$v_6$};
\node [below] at (9,-0.15) {$v_7$};
\node [below] at (10.5,-0.15) {$v_8$};
\node at (1.5,-1.2) {input nodes};
\node at (9,-1.2) {output nodes};
\end{tikzpicture}
\end{center}
\caption{The network $N_5(\{1,3\})$ has $3$ input nodes ($v_1,v_2,v_3$) and $3$ output nodes ($v_6,v_7,v_8$) and $2$ intermediate nodes ($v_4,v_5$).}
\end{figure}
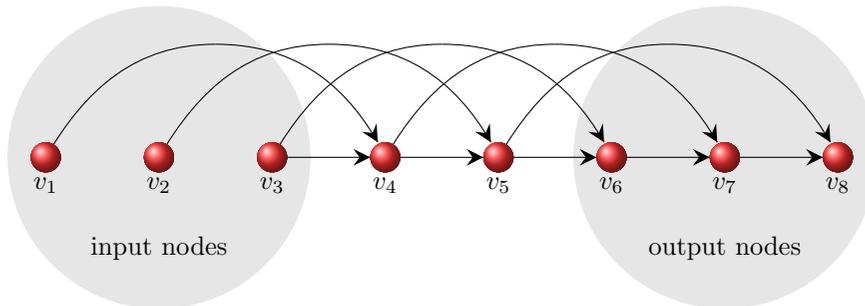

\begin{thm}\cite{Riis2007information}
\label{thm:information_defect}
For any network $N$ of length $n$ and width $r$, if the guessing number of $G_N$ is denoted $\gn(G_N,s)$ and the information defect of $G_N$ is denoted $\bb(G_N,s)$, then 
	$$ \gn(G_N,s) \leq r \qquad \mbox{and} \qquad \bb(G_N,s) \geq n-r. $$
We get the equality $\gn(G_N,s) = r$ if and only if $N$ is $s$-solvable. Moreover, if $N$ is linearly $s$-solvable then $\bb(G_N,s) = n-r$.
\end{thm}

\begin{defn}
\label{defn:clock_network}
For any finite $R \subset \Z^+$ let $r = \max(R)$. For any integer $n>r$ let $N_n(R)$ denote the \emph{clock} network; the network with vertex set is $V = \{ v_1,v_2,v_3, \ldots ,v_{n+r} \}$ and edge set 
	$$ E = \{ v_iv_k \: | \: k>r \mbox{ and } k-i \in R \}. $$
The network $N_n([r])$ is called the \emph{full} clock network. To simplify notation, we sometimes write 
	$$ N_n(r) = N_n([r]). $$
\end{defn}

\begin{defn}
\label{defn:clock_digraph}
For any finite $R \subset \Z^+$ let $r = \max(R)$. For any integer $n>r$ let $G_{clock}(n,R)$ denote the \emph{clock} digraph which has $n$ vertices $\{ v_i \}_{i=1}^n$, where $v_iv_j$ is an edge if and only if $j-i$ (modulo $n$) is in $R$. To simplify notation, for any positive integer $r$, we say 
	$$ G_{clock}(n,r) = G_{clock}(n,[r]). $$
\end{defn}

The clock network inherits it's name from the clock digraph, $G_{clock}(n,R)$, as defined in \cite{Riis2007information}. When $|R| = 2$, the clock digraph is also known as the Cayley graph $\mbox{Cay}(n,R)$, or the ``shift graph'' \cite{Wu2009guessing}. The clock digraph, $G_{clock}(n,R)$, can be obtained from the clock network, $N_n(R)$, by identifying nodes $v_i$ and $v_{n+i}$ for all $1 \leq i \leq r$. \emph{i.e.} 
	$$ G_{clock}(n,R) = G_{N_n(R)}. $$
We show in Theorem~\ref{thm:full_clock} that $N_n(r)$ is always linearly $s$-solvable. By Theorem~\ref{thm:information_defect} (which originally appears in \cite{Riis2007information}) this implies that the guessing number and information defect of $G_{clock}(n,r)$ are $r$ and $n-r$ respectively. 

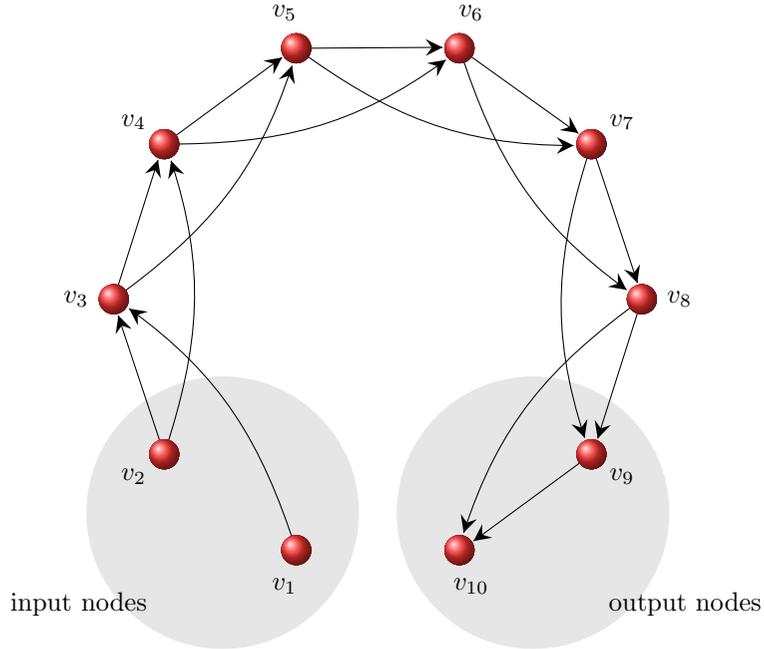
\begin{figure}
\begin{center}
\begin{tikzpicture}
\tikzset{myptr/.style={decoration={markings,mark=at position 1 with 
    {\arrow[scale=2,>=stealth]{>}}},postaction={decorate}}}
\draw[fill=black!10,black!10] (6.5*36:3.5) circle (1.8);
\draw[fill=black!10,black!10] (8.5*36:3.5) circle (1.8);
\draw[myptr] (6*36:3.5) -- (5.1*36:3.45);
\draw[myptr] (5*36:3.5) -- (4.1*36:3.45);
\draw[myptr] (4*36:3.5) -- (3.1*36:3.45);
\draw[myptr] (3*36:3.5) -- (2.1*36:3.45);
\draw[myptr] (2*36:3.5) -- (1.1*36:3.45);
\draw[myptr] (1*36:3.5) -- (0.1*36:3.45);
\draw[myptr] (0*36:3.5) -- (9.1*36:3.45);
\draw[myptr] (9*36:3.5) -- (8.1*36:3.45);
\draw[myptr] (7*36:3.5) to [out=3*36,in=9*36] (5.06*36:3.3);
\draw[myptr] (6*36:3.5) to [out=2*36,in=8*36] (4.06*36:3.3);
\draw[myptr] (5*36:3.5) to [out=1*36,in=7*36] (3.06*36:3.3);
\draw[myptr] (4*36:3.5) to [out=0*36,in=6*36] (2.06*36:3.3);
\draw[myptr] (3*36:3.5) to [out=9*36,in=5*36] (1.06*36:3.3);
\draw[myptr] (2*36:3.5) to [out=8*36,in=4*36] (0.06*36:3.3);
\draw[myptr] (1*36:3.5) to [out=7*36,in=3*36] (9.06*36:3.3);
\draw[myptr] (0*36:3.5) to [out=6*36,in=2*36] (8.06*36:3.3);
\node[style=ball] at (0*36:3.5) {};
\node[style=ball] at (1*36:3.5) {};
\node[style=ball] at (2*36:3.5) {};
\node[style=ball] at (3*36:3.5) {};
\node[style=ball] at (4*36:3.5) {};
\node[style=ball] at (5*36:3.5) {};
\node[style=ball] at (6*36:3.5) {};
\node[style=ball] at (7*36:3.5) {};
\node[style=ball] at (8*36:3.5) {};
\node[style=ball] at (9*36:3.5) {};
\node at (0*36:4) {$v_8$};
\node at (1*36:4) {$v_7$};
\node at (2*36:4) {$v_6$};
\node at (3*36:4) {$v_5$};
\node at (4*36:4) {$v_4$};
\node at (5*36:4) {$v_3$};
\node at (6*36:4) {$v_2$};
\node at (7*36:4) {$v_1$};
\node at (8*36:4) {$v_{10}$};
\node at (9*36:4) {$v_9$};
\node [left] at (6.5*36:5) {input nodes};
\node [right] at (8.5*36:5) {output nodes};
\end{tikzpicture}
\end{center}
\caption{The full clock network $N_8(2)$.}
\end{figure}

\begin{prop}
\label{prop:M_F}
For a given finite $R \subset \Z^+$, let $r = \max(R)$, let $M$ be a $(n+r) \times r$ matrix with entries in $\Z_s$ and for $i=1,2, \ldots ,n+r$, let $\omega(i)$ be the $i^{\mbox{\scriptsize th}}$ row of $M$. If 
\begin{itemize}
\item the first $r$ rows of $M$ form a copy of the identity matrix $I_r$, and
\item for all $r < k \leq n+r$, the row $\omega(k)$ is a linear combination of the rows $\{ \omega(i) \: | \: k-i \in R \}$, 
\end{itemize}
then there exists a circuit, $F$, on $N_n(R)$ such that $M_F=M$.
\end{prop}
\begin{proof}
For $k = r+1,r+2, \ldots ,n+r$, and $j \in R$, let $\lambda_{kj} \in \Z_s$ be the constants by which $\omega(k)$ is a linear combination of $\{ \omega(k-j) \: | \: j \in R \}$. \emph{i.e.} 
	$$ \omega(k) = \sum_{j \in R} \lambda_{kj} \omega(k-j). $$
For all pairs $(k,j)$ such that $k-j \not\in R$ we set $\lambda_{kj}=0$. Now let $F = (f_{r+1},f_{r+2}, \ldots ,f_{n+r})$ be the circuit on $N_n(R)$ defined by 
	$$ X_k = f_k \left( X_i \: \big| \: i \in \Gamma(k) \right) = \sum_{j \in R} \lambda_{kj} X_{k-j}, $$
and for $i=1,2, \ldots ,n+r$, let $\omega_i^\prime$ be the $i^{\mbox{\scriptsize th}}$ row of $M_F$. Since the first $r$ rows of any $R$-circuit matrix form a copy of $I_r$, we must have $\omega(i) = \omega^\prime(i)$ for $i=1,2, \ldots ,r$. Then, inductively, for all $k>r$ we must have 
	$$ \omega(k) = \sum_{j \in R} \lambda_{kj} \omega(k-j) = \sum_{j \in R} \lambda_{kj} \omega^\prime(k-j) = \omega^\prime(k). $$
\end{proof}

\section{Full Clock Networks}

As Theorem~\ref{thm:full_clock} shows, the full clock network is linearly $s$-solvable for all $s$. This is equivalent to Proposition A in \cite{Riis2007information}, however their proof is incomplete (see Example~\ref{ex:Riis_incomplete}). We show that the full clock network is linearly $s$-solvable by finding a valid $[r]$-circuit matrix explicitly. 

\begin{defn}
\label{defn:GIM}
For any integers $a,b>0$ we can define $\I_{a,b}$ in the following recursive manner. If $a=b$, then $\I_{a,a} = I_a$ (the $a \times a$ identity matrix). Otherwise:
	$$ \mbox{if } a<b \mbox{ then } \I_{a,b} = \Big[ \I_{a,b-a} , I_a \Big] , \quad \mbox{and if } a>b \mbox{ then } \I_{a,b} = \left[ \begin{array}{c} \I_{a-b,b} \\ I_b \end{array} \right]. $$ 
So if $a>b$ or $a<b$, then $\I_{a,b}$ is either the horizontal concatenation of $\I_{a,b-a}$ and $I_a$ or the the vertical concatenation of $\I_{a-b,b}$ and $I_b$ respectively. For example: $\I_{4,3}$ and $\I_{30,43}$ are depicted in Figure~\ref{fig:GIM}.
\end{defn}

\begin{figure}
\label{fig:GIM}
  $$ \I_{4,3} = \left[ \begin{array}{ccc} 
  1 & 1 & 1 \\ 
  1 & 0 & 0 \\ 
  0 & 1 & 0 \\ 
  0 & 0 & 1 
    \end{array} \right] \qquad\qquad \I_{30,43} = \tikz[baseline=10ex, scale = 0.1]
  {
  	\draw (0,0) --(43,0) --(43,30) --(0,30) --(0,0); 
	\draw (13,0) --(13,30);
	\draw (0,13) --(13,13);
	\draw (0,26) --(13,26);
	\draw (9,26) --(9,30);
	\draw (5,26) --(5,30);
	\draw (1,26) --(1,30);
	\draw (0,27) --(1,27);
	\draw (0,28) --(1,28);
	\draw (0,29) --(1,29);
	\node at (6.5,6.5) {$I_{13}$};
	\node at (6.5,19.5) {$I_{13}$};
	\node at (28,15) {$I_{30}$};
  }
  $$
\caption{The matrices $\I_{4,3}$ and $\I_{30,43}$. Each square represents a copy of an identity matrix.}
\end{figure}

\begin{prop}
\label{prop:detA=1}
If $A$ is the topleft-most $a \times a$ sub-matrix of $\I_{n,r}$, then $|\det(A)|=1$. 
\end{prop}
\begin{proof}
Let $A$ be the topleft-most $a \times a$ submatrix of $\I_{n,r}$. We now construct the pair of integers $p$ and $q$ in the following way. Initially let $x=n$ and $y=r$. Then iteratively perform the following process.
	\begin{align*} 
	\mbox{while } & x>a \mbox{ or } y>a: \\
	& \mbox{if $x>y$ replace $x$ with $x-y$,} \\
	& \mbox{otherwise replace $y$ with $y-x$.} 
	\end{align*}
Throughout this process (by Definition~\ref{defn:GIM}) topleft-most $x \times y$ submatrix of $A$ is always a copy of $\I_{x,y}$. As soon as both $x$ and $y$ are less than or equal to $a$, we set $p=x$ and $q=y$, and terminate this process. Just before the final iteration, we must have had one of $x$ or $y$ greater than $a$, so $a < \max(x,y) = p+q$. Now $A$ must be in the following form. 
	\begin{equation}
	\label{eq:A}
	A = \left[ \begin{array}{cc} \I_{p,q} & P \\ Q & S \end{array} \right] 
	\end{equation}
where $P$, $Q$ and $S$ are matrices with dimensions $p \times (a-q)$, $(a-p) \times q$ and $(a-p) \times (a-q)$ respectively. Now, there are two cases: 
\begin{itemize}
\item If $P$ is the left-most $(a-q)$ columns of a copy of $I_p$, then $[Q,S]$ is the topleft-most $(a-p) \times a$ submatrix of a large identity matrix.
\item If $Q$ is the top-most $(a-p)$ rows of a copy of $I_q$, then $\left[ \begin{array}{c} P \\ S \end{array} \right]$ is the topleft-most $a \times (a-q)$ submatrix of a large identity matrix.
\end{itemize}
In either case, all the entries of $S$ must be zero because $a-p<q$ and $a-q<p$. Now consider $\I_{p,q}$ in which the bottomright-most square submatrix must be a copy of an identiy matrix. Explicitly, for any integer $b$ such that $0 \leq b \leq \min(p,q)$, we have 
	\begin{equation}
	\label{eq:Ipq}
	\I_{p,q} = \left[ \begin{array}{cc} * & * \\ * & I_b \end{array} \right] , 
	\end{equation}
where $*$ denotes arbitrary entries. In particular $\I_{p,q}$ has this form for $b=p+q-a$. Substituting Equation~(\ref{eq:Ipq}) into Equation~(\ref{eq:A}), we see that $A$ has the form:
	$$ A = \left[ \begin{array}{cc} \I_{p,q} & P \\ Q & 0 \end{array} \right] = \left[ \begin{array}{ccc} * & * & I_{a-q} \\ * & I_{p+q-a} & 0 \\ I_{a-p} & 0 & 0 \end{array} \right] $$
where a $0$ denotes a submatrix full of zeros, and $*$ denotes a submatrix with arbitrary entries. The only non-zero terms in the Leibniz formula for the determinant of $A$ must come exclusively from the submatrices labelled $I_{a-q}$, $I_{p+q-a}$ and $I_{a-p}$. Therefore 
  $$ |\det(A)| = |\det(I_{a-q})| \times |\det(I_{p+q-a})| \times |\det(I_{a-p})| = 1. $$
\end{proof}

\begin{figure}
\begin{center}
\begin{tikzpicture}[baseline=10ex, scale = 0.3]
  \draw[ultra thick, <->] (0,0) -- (0,9) -- (13,9); 
  \draw[thick] (0,2) -- (7,2) -- (7,9);
  \draw [fill=yellow,ultra thick] (0,4) rectangle (6,9);
  \draw (6,4) -- (11,4) --(11,9);
  \node at (3,6.5) {$\I_{p,q}$};
  \node at (8.5,6.5) {$I_p$};
  \draw[<->] (-1,2) -- (-1,9); 
  \node[left] at (-1,5.5) {$a$};
  \draw[<->] (0,10) -- (7,10); 
  \node[above] at (3.5,10) {$a$};
\end{tikzpicture}
$\qquad$ or $\qquad$
\begin{tikzpicture}[baseline=10ex, scale = 0.3]
  \draw[ultra thick, <->] (0,-3) -- (0,9) -- (10,9); 
  \draw[thick] (0,2) -- (7,2) -- (7,9);
  \draw [fill=yellow,ultra thick] (0,4) rectangle (6,9);
  \draw (6,4) -- (6,-2) -- (0,-2);
  \node at (3,6.5) {$\I_{p,q}$};
  \node at (3,1) {$I_q$};
  \draw[<->] (-1,2) -- (-1,9); 
  \node[left] at (-1,5.5) {$a$};
  \draw[<->] (0,10) -- (7,10); 
  \node[above] at (3.5,10) {$a$};
\end{tikzpicture}
\end{center}
\caption{In the proof of Proposition~\ref{prop:detA=1}, $\I_{p,q}$ is the topleft-most $p \times q$ submatrix of $A$ and $A$ is the topleft-most $a \times a$ submatrix of $\I_{n.r}$.}
\label{fig:detA=1}
\end{figure}
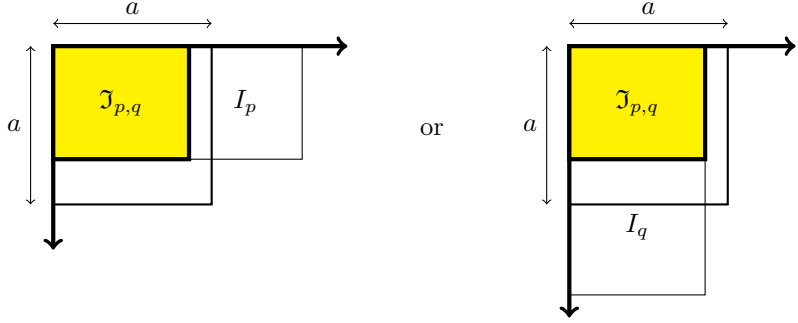

\begin{defn}
\label{defn:M}
For any positive integers $n$ and $r$ with $n>r$ we define the $(n+r) \times r$ matrix $\M_{n,r}$ formed by concatenating a copy of $I_r$ ontop of $\I_{n,r}$. \emph{i.e.} 
	$$ \M_{n,r} = \left[ \begin{array}{c} I_r \\ \I_{n,r}  \end{array} \right]. $$
\end{defn}

\begin{prop}
\label{prop:independent_rows}
For any positive integers $n>r$, if $M$ is an $r \times r$ submatrix of $\M_{n,r}$ formed by $r$ consecutive rows then $|\det(M)| = 1$.
\end{prop}
\begin{proof}
Let $M$ be the rows $\omega_{a+1}, \omega_{a+2}, \ldots ,\omega_{a+r}$. There are two cases: either $0 \leq a < r$ or $r \leq a \leq n$.
\begin{itemize}
\item If $a<r$ then $M$ consists of the final $r-a$ rows of $I_r$ followed by the initial $a$ rows of $\I_{n,r}$. In this case, $M$ must have the form:
	$$ M = \left[ \begin{array}{cc} 0 & I_{r-a} \\ A & * \end{array} \right] $$
Where $A$ is the top-leftmost $a \times a$ submatrix of $\I_{n,r}$ and $*$ denotes a submatrix with arbitrary entries. In this case, by Proposition~\ref{prop:detA=1}, we have 
	$$ |\det(M)| = |\det(I_{r-a})| \times |\det(A)| = 1 $$
\item If $a \geq r$ then $M$ must have the form: 
  $$ M = \left[ \begin{array}{cc} * & I_b \\ I_{r-b} & 0 \end{array} \right] $$
where $*$ denotes a submatrix with arbitrary entries and $b$ is the remainder when $n-a$ is divided by $r$. In this case we have 
	$$ |\det(M)| = |\det(I_{r-b})| \times |\det(I_b)| = 1. $$
\end{itemize}
\end{proof}

\begin{thm}
\label{thm:full_clock}
For any $n \geq r > 0$ and any $s \geq 2$, the full clock network, $N_n(r)$, is linearly $s$-solvable. 
\end{thm}
\begin{proof}
Consider the matrix $\M_{n,r}$ as defined in Definition~\ref{defn:M}. By Proposition~\ref{prop:independent_rows}, for any $s$, the integer span of any $r$ consecutive rows of $\M_{n,r}$ is all $\Z_s^r$. So any row can be expressed as a linear combination of the preceding $r$ rows. Moreover, the first $r$ rows of $\M_{n,r}$ form a copy of $I_r$. Therefore $\M_{n,r}$ satisfies the conditions of Proposition~\ref{prop:M_F}, and so there is a circuit $F$ on $N_n(r)$ such that $M_F=\M_{n,r}$. This circuit linearly solves $N_n(r)$ because the final $r$ rows of $\M_{n,r}$ form a copy of $I_r$. 
\end{proof}

Example~\ref{ex:Riis_incomplete} is the construction used in the incomplete proof of Proposition A in \cite{Riis2007information}. For certain integers $n$ and $r$, this construction does not solve $N_n(r)$. 

\begin{ex}
\label{ex:Riis_incomplete}
Let $F$ be a circuit on $N_n(r)$ over $\Z_s$ such that the valuation of $F$ satisfies
	$$ X_i+X_{i+1}+\cdots +X_{i+r} \equiv 0 \; (\mbox{mod } s) \qquad \mbox{ for } i=1,2,3, \ldots ,n-r $$
for any input $c \in \Z_s^r$. To see that this circuit does not solve $N_n(r)$ in general, observe that for $n=7$ and $r=2$, it does not solve $N_7(2)$. Explicitly, for any $c = (c_1,c_2) \in \Z_s^2$, the valuation must satisfy: 
	\begin{align*}
	X_1 & = c_1 \\
	\mbox{and} \quad X_2 & = c_2, 
\intertext{because for any valuation of a circuit on $N_7(2)$, we have $(X_1,X_2) = c$. Moreover, for $j=3,4,5,6$ and $7$, we can deduce:}
	X_3 & = -c_1-c_2, \\
	X_4 & = c_1, \\
	X_5 & = c_2, \\
	X_6 & = -c_1-c_2 \\
	\mbox{and} \quad X_7 & = c_1, 
\intertext{because $X_{j-2}+X_{j-1}+X_j \equiv 0$. Finally, if $F$ solved $N_7(2)$, then we would have:}
	X_8 & = c_1 \\
	\mbox{and} \quad X_9 & = c_2. 
	\end{align*}
However, this is not possible; $X_9=c_2$ cannot be determined from only $X_7=c_1$ and $X_8=c_1$. 
\end{ex}

\section{Analysis of a specific case}
\label{sec:analysis_of_R=1,3}

\begin{figure}
$$ \begin{array}{r|ccccccccc} n & 4 & 5 & 6 & 7 & 8 & 9 & 10 & 11 & \geq 12 \\ \hline 
\mbox{Is $N_n(\{1,3\})$ linearly $2$-solvable?} 
 & \mbox{\color{red} \ding{53}}  & \mbox{\color{red} \ding{53}} & \mbox{\color{green} \ding{51}} 
 & \mbox{\color{green} \ding{51}} & \mbox{\color{red} \ding{53}} & \mbox{\color{green} \ding{51}} 
 & \mbox{\color{green} \ding{51}} & \mbox{\color{red} \ding{53}} & \mbox{\color{green} \ding{51}} \\
\mbox{Is $N_n(\{1,3\})$ linearly $3$-solvable?} 
 & \mbox{\color{red} \ding{53}}  & \mbox{\color{red} \ding{53}} & \mbox{\color{green} \ding{51}} 
 & \mbox{\color{red} \ding{53}} & \mbox{\color{green} \ding{51}} & \mbox{\color{green} \ding{51}} 
 & \mbox{\color{green} \ding{51}} & \mbox{\color{green} \ding{51}} & \mbox{\color{green} \ding{51}} \\\end{array} $$
\caption{The linear $2$-solvability and linear $3$-solvability of $N_n(\{1,3\})$ for all $n \geq 4$.}
\end{figure}

In this section we investigate the $2$-solvability and $3$-solvability of the network $N_n(\{1,3\})$ for various values $n$. Firstly, we consider $n=7$ and $n=8$ in the following example.

\begin{ex}
Let $n=7$, $m=3$ and $R = \{ 1,3 \}$ and consider the following two matrices $A$ and $B$ defined as follows. 
  $$ A = \left[ \begin{array}{ccc}
  1 & 0 & 0 \\ 
  0 & 1 & 0 \\ 
  0 & 0 & 1 \\
  1 & 0 & 1 \\
  1 & 1 & 1 \\
  1 & 1 & 0 \\
  0 & 1 & 1 \\
  1 & 0 & 0 \\
  0 & 1 & 0 \\
  0 & 0 & 1 \end{array} \right] \qquad \mbox{and} \qquad  
  B = \left[ \begin{array}{ccc}
  1 & 0 & 0 \\ 
  0 & 1 & 0 \\ 
  0 & 0 & 1 \\
  1 & 0 & 1 \\
  1 & 1 & 1 \\
  1 & 1 & 2 \\
  2 & 1 & 0 \\
  0 & 2 & 1 \\
  1 & 0 & 0 \\
  0 & 1 & 0 \\
  0 & 0 & 1 \end{array} \right] $$ 
The matrix $A$ is constructed so that the $i^{\mbox{\scriptsize th}}$ row is the sum of the $(i-1)^{\mbox{\scriptsize th}}$ row and the $(i-3)^{\mbox{\scriptsize th}}$ row modulo $2$. So, by Proposition~\ref{prop:M_F}, $A$ is a valid $\{1,3\}$-circuit matrix (over $\Z_2$) and since the bottom $3$ rows of $A$ form an identity matrix, this demonstrates that $N_7(\{1,3\})$ is linearly $2$-solvable. Similarly the matrix $B$ demonstrates that $N_8(\{1,3\})$ is linearly $3$-solvable. It can be verified by a brute force computer search that $N_7(\{1,3\})$ is not linearly $3$-solvable and $N_8(\{1,3\})$ is not linearly $2$-solvable.
\end{ex}

In general, the $s$-solvability of a network depends on $s$. However, for any $n \geq 12$, we can construct a $\{1,3\}$-circuit matrix of length $n$ which is valid over $Z_s$ for any $s \geq 2$ in the following way. For $n \equiv 0,1,2 \: (\mbox{mod } 3)$ iteratively concatenate copies of $I_3$ to the bottom of $I_3$, $M_{10}$ or $M_{14}$ respectively, where $M_{10}$ and $M_{14}$ are given in Figure~\ref{fig:M10andM14}. We know that no such $\{1,3\}$-circuit matrix exists for $n=7$ nor $n=11$ because (by brute force computer search) we computed that $N_7(\{1,3\})$ is not linearly $3$-solvable and $N_{11}(\{1,3\})$ is not $2$-solvable.

\begin{figure}
$$ M_{10} = \left[ \begin{array}{ccc}
  1 & 0 & 0 \\ 
  0 & 1 & 0 \\ 
  0 & 0 & 1 \\
  1 & 0 & 1 \\ 
  1 & 1 & 1 \\ 
  0 & 0 & 1 \\ 
  -1 & 0 & 1 \\ 
  1 & 1 & 1 \\ 
  1 & 1 & 0 \\ 
  0 & 1 & 1 \\ 
  1 & 0 & 0 \\
  0 & 1 & 0 \\
  0 & 0 & 1 \end{array} \right] \qquad \mbox{ and } \qquad
  M_{14} = \left[ \begin{array}{ccc}
  1 & 0 & 0 \\ 
  0 & 1 & 0 \\ 
  0 & 0 & 1 \\
  1 & 0 & 0 \\ 
  1 & 1 & 0 \\ 
  1 & 1 & 1 \\ 
  0 & 1 & 1 \\ 
  1 & 0 & -1 \\ 
  0 & 1 & 2 \\ 
  0 & 0 & 1 \\ 
  1 & 0 & -1 \\ 
  1 & 1 & 1 \\ 
  1 & 1 & 0 \\ 
  0 & 1 & 1 \\ 
  1 & 0 & 0 \\
  0 & 1 & 0 \\
  0 & 0 & 1 \end{array} \right] $$ 
\caption{Matrices $M_{10}$ and $M_{14}$ show that $N_{10}(\{1,3\})$ and $N_{14}(\{1,3\})$ are linearly $s$-solvable for any $s \geq 2$.}
\label{fig:M10andM14}
\end{figure}

\section{General Clock Networks}

We saw in the previous section that the network $N_n(\{1,3\})$ is linearly $s$-solvable for any $s$, for all $n \geq 12$. In this section we generalise this result to arbitrary finite sets of positive integers. Specifically, we determine for which finite $R \subset \Z^+$, does there exist a constant $n_0$ such that $N_n(R)$ is $s$-solvable for all $s$ and all $n \geq n_0$. We deduce (by Lemma~\ref{lem:gcd>1} and Corollary~\ref{cor:large_n},) that such an integer $n_0$ exists if and only if $\gcd(R)=1$. 

\begin{lem}
\label{lem:gcd>1}
If $n$ is not a multiple of $\gcd(R)$, then $N_n(R)$ is not $s$-solvable for any $s \geq 2$. 
\end{lem}
\begin{proof}
Let $d=\gcd(R)>1$. By definition, each edge $v_iv_k$ only joins vertices such that $i \equiv k \; (\mbox{mod } d)$. Therefore $N_n(R)$ is disconnected with at least one component for each residue modulo $d$. Now consider some index $a$, and an input $c = (c_1,c_2, \ldots ,c_r)$. If we keep $c_i$ constant for all $i \not= a$ and let $c_a$ vary, then the valuation will only change on vertices in the same component as $v_a$. Since $n$ is not a multiple of $d$, 
	$$ (n+a) - a = n \not\equiv 0 \: (\mbox{mod } d). $$ 
So the input node, $v_a$, and its corresponding output node, $v_{n+a}$, are in a different components of $N_n(R)$. Therefore $N_n(R)$ is not $s$-solvable for any $s \geq 2$.
\end{proof}

Now consider any $R \subset \Z^+$ such that $\gcd(R)>1$, there are an infinite number of integers $n$ which are not a multiple of $\gcd(R)$. By Lemma~\ref{lem:gcd>1}, this is an infinite number of integers $n$ such that $N_n(R)$ is not $s$-solvable (for any $s$). Therefore there cannot exist any $n_0$ such that $N_n(R)$ is $s$-solvable for all $n \geq n_0$. However, if $n$ is a multiple of $d = \gcd(R) > 1$, then the network $N_n(R)$ is a disjoint union of $d$ copies of 
	$$ N^\prime = N_{n/d}(R^\prime) \quad \mbox{where} \quad R^\prime = \{ j/d \: | \: j \in R \}. $$
So $N_n(R)$ is $s$-solvable if and only if $N^\prime$ is $s$-solvable. Since $\gcd(R^\prime)=1$, it suffices now to consider only the cases that $\gcd(R)=1$. We now make the following definition and propositions, used in Theorem~\ref{thm:large_n} and Corollary~\ref{cor:large_n}.

\begin{defn}
Let $s \geq 2$ be an integer, let $R$ be a finite set of positive integers, and let $r = \max(R)$.
\begin{itemize}
\item An \emph{$R$-atomic} matrix is any $r \times r$ matrix, with entries in $\Z_s$, of the form:
	$$ \left[ \begin{array}{ccccc} 
	0 & 1 & 0 & \cdots & 0 \\
	0 & 0 & 1 & \cdots & 0 \\
	\vdots & \vdots & \vdots & \ddots & \vdots \\
	0 & 0 & 0 & \cdots & 1 \\
	\alpha_r & \alpha_{r-1} & \alpha_{r-2} & \cdots & \alpha_1 
	\end{array} \right] $$
such that $\alpha_j = 0$ for all $j \not\in R$. 
\item A \emph{$R$-step} matrix is any $r \times r$ matrix, with entries in $\Z_s$, formed by starting with $I_r$, and then for some $1 \leq t \leq r$, replacing the $t^{\mbox{\scriptsize th}}$ row with $[\beta_1,\beta_2, \ldots ,\beta_r]$, 
	$$ \left[ \begin{array}{cccccc} 
	1 & 0 & \cdots & 0 & \cdots & 0 \\
	0 & 1 & \cdots & 0 & \cdots & 0 \\
	\vdots & \vdots & \ddots & \vdots & & \vdots \\
	\beta_1 & \beta_2 & \cdots & \beta_t & \cdots & \beta_r \\
	\vdots & \vdots & & \vdots & \ddots & \vdots \\
	0 & 0 & \cdots & 0 & \cdots & 1 \\
	\end{array} \right] $$
where $\beta_i$ is non-zero only if there is some $j \in R$ such that $i+j \equiv t \: (\mbox{mod } r)$. Since $r \in R$, we always allow $\beta_t$ to be non-zero.
\item For any $1 \leq t \leq r$, the  \emph{$t$-toggle} matrix is the following $r \times r$ matrix, $T(t)$, with entries in $\Z_s$.
	$$ T(t) = \left[ \begin{array}{cccccc} 
	1 & 0 & \cdots & 0 & \cdots & 0 \\
	0 & 1 & \cdots & 0 & \cdots & 0 \\
	\vdots & \vdots & \ddots & \vdots & & \vdots \\
	-1 & -1 & \cdots & -1 & \cdots & -1 \\
	\vdots & \vdots & & \vdots & \ddots & \vdots \\
	0 & 0 & \cdots & 0 & \cdots & 1 \\
	\end{array} \right]. $$
\emph{i.e.} The $t$-toggle matrix is formed from $I_r$ by replacing row $t$ with a row of $-1$s. A matrix is called a \emph{toggle} matrix iff it is a $t$-toggle matrix for some $t$.
\end{itemize}
\end{defn}

\begin{prop}
\label{prop:atomic2step}
Any $R$-step matrix can be expressed as a product of $r$ $R$-atomic matrices. 
\end{prop}
\begin{proof}
Let $P$ denote the only $R$-atomic matrix which is also a permutation matrix; the $R$-atomic matrix for which $\alpha_r=1$ and $\alpha_i = 0$ for all $i<r$. For any $1 \leq t \leq r$ consider the product 
	$$ A_rA_{r-1} \ldots A_t \ldots A_2A_1 = S, $$ 
where $A_t$ is an arbitrary $R$-atomic matrix and $A_i=P$ for all $i \not= t$. This product, $S$, is an arbitrary $R$-step matrix. 
\end{proof}

\begin{prop}
\label{prop:step2toggle}
If $\gcd(R)=1$ and $r = \max(R) > 1$ then for any $1 \leq t \leq r$ the $t$-toggle matrix can be expressed as a product of $(2r-3)$ $R$-step matrices.
\end{prop}
\begin{proof}
We inductively define a sequence of subsets, 
	$$ U_2 \subset U_3 \subset U_4 \subset \cdots \subset U_r = [r], $$ 
in the following manner. Let $U_2 = \{ x , t \}$ where $x \in [r]$ is chosen so that $t-x \; (\mbox{mod } r) \in R$. For $k = 3,4,5, \ldots,r$, iteratively define $U_k = U_{k-1} \cup \{ b \}$ for some $b \in [r] \backslash U_{k-1}$ such that there exists some $a \in U_{k-1}$ such that $a-b \; (\mbox{mod } r) \in R$. We know $a$ and $b$ exist because $\gcd(R)=1$. Now let $S_k$ be the matrix formed from $I_r$ by replacing the $t^{\mbox{\scriptsize th}}$ row with 
	$$ (x_1,x_2,x_3, \ldots ,x_n) \quad \mbox{where} \quad 
	x_i = \begin{cases} -1 & : \mbox{ if } i \in U_k \\ 0 & : \mbox{ otherwise}. \end{cases} $$ 
Now we prove that $S_k$ can be expressed as a product of $(2k-3)$ $R$-step matrices, by induction on $k=2,3,4, \ldots ,r$. For the base case ($k=2$), $S_2$ is a $R$-step matrix. For the inductive step, 
	$$ S_k = E_{ab}(-1) S_{k-1} E_{ab}(1) $$
where $E_{ij}(\lambda)$ is the matrix formed from $I_r$ by replacing the $ij^{\mbox{\scriptsize th}}$ entry with $\lambda$. Note that $E_{ab}(1)$ and $E_{ab}(-1)$ are both $R$-step matrices because $a-b \: (\mbox{mod } r) \in R$, and $S_{k-1}$ can be expressed as a product of $(2(k-1)-3)$ $R$-step matrices by the inductive assumption. Therefore $S_k$ can be expressed as a product of 
	$$ 1 + (2(k-1)-3) + 1 = 2k-3 \;\;\; \mbox{ $R$-step matrices.} $$ 
This completes the induction. For $k=r$ we have $U_r = \{ 1,2, \ldots ,r \}$ and so the $t$-toggle matrix is $T(t)=S_r$, which can be expressed as a product of $(2r-3)$ $R$-step matrices.
\end{proof}

\begin{prop}
\label{prop:toggle2permutation}
Any $r \times r$ permutation matrix can be expressed as the product of at most $\frac{3r}{2}$ toggle matrices.
\end{prop}
\begin{proof}
First we show that an arbitrary $k$-cycle can be expressed as the product of $k+1$ toggle matrices. Explicitly, if $Q$ is the $r \times r$ matrix corresponding to the $k$ cycle, $(a_1,a_2, \ldots ,a_k)$, can be expressed as the product 
	$$ Q = T(a_1) T(a_2) T(a_3) \cdots T(a_{k-1}) T(a_k) T(a_1). $$
Now consider the cyclic decomposition of the permutation; the permutation expressed as the composition of at most $n/2$ cycles, such that the sum of the lengths of these cycles is at most $n$. If each of these cycles are expressed as a product of toggle matrices, then this is 
\end{proof}

\begin{thm}
\label{thm:large_n}
For any finite $R \subset \Z^+$, the network $N_n(R)$ is linearly $s$-solvable if and only if the identity matrix can be expressed as a product of $n$ $R$-atomic matrices with entries in $\Z_s$.
\end{thm}
\begin{proof}
For any valuation, $X = (X_1,X_2, \ldots ,X_n)$, of a linear circuit on $N_n(R)$, let $Y_i$ denote the column vector $Y_i = (X_{i+1},X_{i+2},X_{i+3}, \ldots ,X_{i+r})^T$ for $i=0,1,2, \ldots ,n$. Since $f_k$ is linear, 
	$$ X_k = f_k\left( X_{k-j} \: | \: j \in R \right) = \sum_{j \in R} \alpha_{kj} X_{k-j}. $$
Therefore we must have $Y_i = A_iY_{i-1}$ where $A_i$ is a $R$-atomic matrix. Inductively this implies that $Y_i = A_iA_{i-1} \ldots A_2A_1Y_0$ for all $i \geq 0$ and thus $Y_n = \A Y_0$, where $\A = A_nA_{n-1} \cdots A_2A_1$ is a product of $n$ $R$-atomic matrices. For all inputs $c \in \Z_s^r$ we have 
	$$ F(c) = Y_n = \A Y_0 = \A c. $$
So if $X$ is the valuation of a circuit which linearly $s$-solved $N_n(R)$, then $c = \A c$ for all $c \in \Z_s^r$. Hence $\A=I_r$, and $\A$ is a product of exactly $n$ $R$-atomic matrices. Conversely the function $f_k$ can be reconstructed from the $R$-atomic matrix $A_{k-r}$, for each $k=r+1,r+2, \ldots ,n+r$, so this construction is reversible. 
\end{proof}
\begin{cor}
\label{cor:large_n}
Let $R$ be any finite set of positive integers with $\gcd(R) = 1$ and let $r = \max(R)$. For any $n \geq = 3r^3$, the network $N_n(R)$ is linearly $s$-solvable for any integer $s \geq 2$.
\end{cor}
\begin{proof}
It suffices to show that the identity matrix can be expressed as a product of exactly $n$ $R$-atomic matrices. Let $P$ denote the only atomic matrix which is also a permutation matrix; the atomic matrix for which $\alpha_r=1$ and $\alpha_i = 0$ for all $i<r$. Note that $P$ is a $r$-cycle and so $P^r = I_r$. Let $Q=P^{-n}$ and note that $Q$ is a permutation matrix. By Proposition~\ref{prop:toggle2permutation}, we can write $Q$ as a product of $k \leq \frac{3r}{2}$ toggle matrices. By Proposition~\ref{prop:step2toggle}, we can write each of these toggle matrices as a product of $(2r-3)$ $R$-step matrices, and by Proposition~\ref{prop:atomic2step} we can write each of these step matrices as a product of $r$ $R$-atomic matrices. Therefore $Q$ can be expressed as a product of $kr(2r-3)$ $R$-atomic matrices. Since $k \leq \frac{3r}{2}$ and $n \geq 3r^3 > \tfrac{3r}{2} r(2r-3)$, we must have $n-kr(2r-3) \geq 0$ and so 
	$$ Q \times P^{n-kr(2r-3)} = \left( Q \times P^n \right) \times \left( P^r \right)^{k(2r-3)} = I_r. $$
Since $P$ is a $R$-atomic matrix and $Q$ can be expressed as a product of $kr(2r-3)$ $R$-atomic matrices, we can express $I_r$ as a product of $n$ $R$-atomic matrices.
\end{proof}

For finite $R \subset \Z^+$ with $\gcd(R)=1$, let $n_0 = n_0(R)$ be the minimum integer such that $N_n(R)$ is $s$-solvable for all $s \geq 2$ for all $n \geq n_0$. Corollary~\ref{cor:large_n} shows that $n_0$ is well defined and that $n_0 \leq 3r^3$ (where $\max(R)=r$). Theorem~\ref{thm:full_clock} shows that $n_0([r]) = r$ and in Section~\ref{sec:analysis_of_R=1,3}, we deduced that $n_0(\{1,3\}) = 12$. The value $n_0(R)$ for $R$ in general, remains an open question. We conclude this section with an example which demonstrates that the cubic bound $n_0 \leq 3r^3$ in Corollary~\ref{cor:large_n} cannot be replaced with any bound less than $r^2-r$. 

\begin{ex}
For any integer $r \geq 3$, let $n = r^2-r-1$ and consider the networks $N = N_n(\{1,r\})$ and $N^\prime = N_n(\{1,r-1\})$, and consider the digraphs $G = G_{clock}(n,\{1,r\})$ and $G^\prime = G_{clock}(n,\{1,r-1\})$. Using Theorem~\ref{thm:information_defect}, the existence of the acyclic network $N^\prime$ of width $r-1$ such that $G^\prime = G_{N^\prime}$ implies that 
	$$ \gn(G^\prime,s) \leq r-1. $$
Note that $r(r-1) \equiv 1$ modulo $n$, so $G$ and $G^\prime$ are isomorphic, and so $\gn(G,s) = \gn(G^\prime,s) < r$. Now using Theorem~\ref{thm:information_defect} again (since $G = G_N$) we can conclude that $N$ is not $s$-solvable. Thus 
	$$ n_0(\{1,r\}) \geq r^2-r, $$
for all $r \geq 3$.
\end{ex}

\bibliography{information_flow}
\bibliographystyle{plain}

\end{document}